\documentclass[journal]{IEEEtran}

\usepackage[cmex10]{amsmath}
\usepackage{amssymb}
\usepackage[numbers,sort&compress]{natbib}
\usepackage{algorithm,algpseudocode}
\usepackage[short]{optidef}
\usepackage{siunitx}
\usepackage{xcolor}
\usepackage{commath}  

\newcommand{\rbr}[2][-1]{\del[#1]{#2}}           
\newcommand{\defn}{\coloneqq}                    
\newcommand{\R}{\mathbb{R}}                      
\newcommand{\tran}[1]{#1^\mathsf{T}}             
\DeclareMathOperator{\grad}{\nabla\!}            
\newcommand{\x}{\times}                          

\newcommand\mydots{\hbox to 1em{.\hss.\hss.}} 

\newcommand\AverageSmallMatrix[1]{{%
  \footnotesize\arraycolsep=0.22\arraycolsep\ensuremath{\begin{bmatrix}#1\end{bmatrix}}}}

\newcommand{\Lagr}{\mathcal{L}}

\newtheorem{theorem}{Theorem}

\newtheorem{assumption}{Assumption}
\newtheorem{definition}{Definition}

\newtheorem{remark}{Remark}

\begin{document}

\title{Control Barrier Functions for Cyber-Physical Systems and Applications to NMPC}

\author{%
    Jan~Schilliger\textsuperscript{1},
    Thomas~Lew\textsuperscript{2},
    Spencer~M.~Richards\textsuperscript{2},
    Severin~H\"{a}nggi\textsuperscript{1},
    Marco~Pavone\textsuperscript{2},
    and~Christopher~Onder\textsuperscript{1}%
    \thanks{\textsuperscript{1}Jan Schilliger, Severin H\"{a}nggi, and Christopher Onder are with the Institute for Dynamic Systems and Control (IDSC), ETH Z\"{u}rich, Z\"{u}rich, Switzerland. \{janschi, shaenggi, onder\} @ethz.ch}%
    \thanks{\textsuperscript{2}Thomas Lew, Spencer~M.\ Richards, and Marco Pavone are with the Department of Aeronautics and Astronautics, Stanford University, Stanford, CA 94305. \{thomas.lew, spenrich, pavone\} @stanford.edu}%
    \thanks{Jan Schilliger is partially supported by the Master's Thesis Grant of the Zeno Karl Schindler Foundation for his work at the Autonomous Systems Lab, Stanford University.}%
}


\maketitle

\begin{abstract}\boldmath\textbf{%
   Tractable safety-ensuring algorithms for cyber-physical systems are important in critical applications. Approaches based on Control Barrier Functions assume continuous enforcement, which is not possible in an online fashion. This paper presents two tractable algorithms to ensure forward invariance of discrete-time controlled cyber-physical systems. Both approaches are based on Control Barrier Functions to provide strict mathematical safety guarantees. The first algorithm exploits Lipschitz continuity and formulates the safety condition as a robust program which is subsequently relaxed to a set of affine conditions. The second algorithm is inspired by tube-NMPC and uses an affine Control Barrier Function formulation in conjunction with an auxiliary controller to guarantee safety of the system. We combine an approximate NMPC controller with the second algorithm to guarantee strict safety despite approximated constraints and show its effectiveness experimentally on a mini-Segway.%
}\end{abstract}

\begin{IEEEkeywords}
    Optimization and Optimal Control, Robot safety, control barrier functions, nonlinear model predictive control
\end{IEEEkeywords}

\IEEEpeerreviewmaketitle

\section{Introduction}\label{sec:intro}




Two cornerstones of safety-critical control for robotic systems are ``stability'' (i.e., convergence towards desired behavior) and ``safety'' (i.e., remaining within a designated set of safe states). A key challenge to deployment of cyber-physical systems is the design of fast, tractable control algorithms with \emph{safety guarantees} that hold in the face of real-world challenges such as discretization error. For instance, as shown in Figure~\ref{fig:mini-segway_1}, a mini-Segway must maintain a bounded pitch angle at all times to avoid falling over.


Nonlinear Model Predictive Control (NMPC) is a promising method for safety-critical control. In NMPC, a control input is obtained as the first input of an optimal trajectory computed by a constrained optimization at each time step. Indeed, the NMPC problem can encode safety with state and input constraints, but these are generally only enforced at discrete time steps. One approach to obtain continuous-time constraint enforcement is to account for discretization error in tube NMPC \cite{kogel2015discrete}. However, NMPC is computationally demanding and thus remains an open problem for systems with pronounced nonlinearities, high dimensionality, or fast response times \cite{liao2019time}. To address this, an NMPC problem can be replaced with an approximation that is easier to solve \cite{desaraju2016fast,neunert2016fast,zanelli2019efficient}. For example, the Real-Time Iteration (RTI) scheme approximates an NMPC problem with a Quadratic Program (QP) at each time-step \cite{diehl2005real}. While RTI has been shown to be stabilizing, strict constraint satisfaction, recursive feasibility, and hence safety are no longer guaranteed, due to linearly approximated constraints \cite{zanelli2016efficient,zanelli2019efficient}.

\begin{figure}[t]
	\centering
	\includegraphics[width=0.9\linewidth]{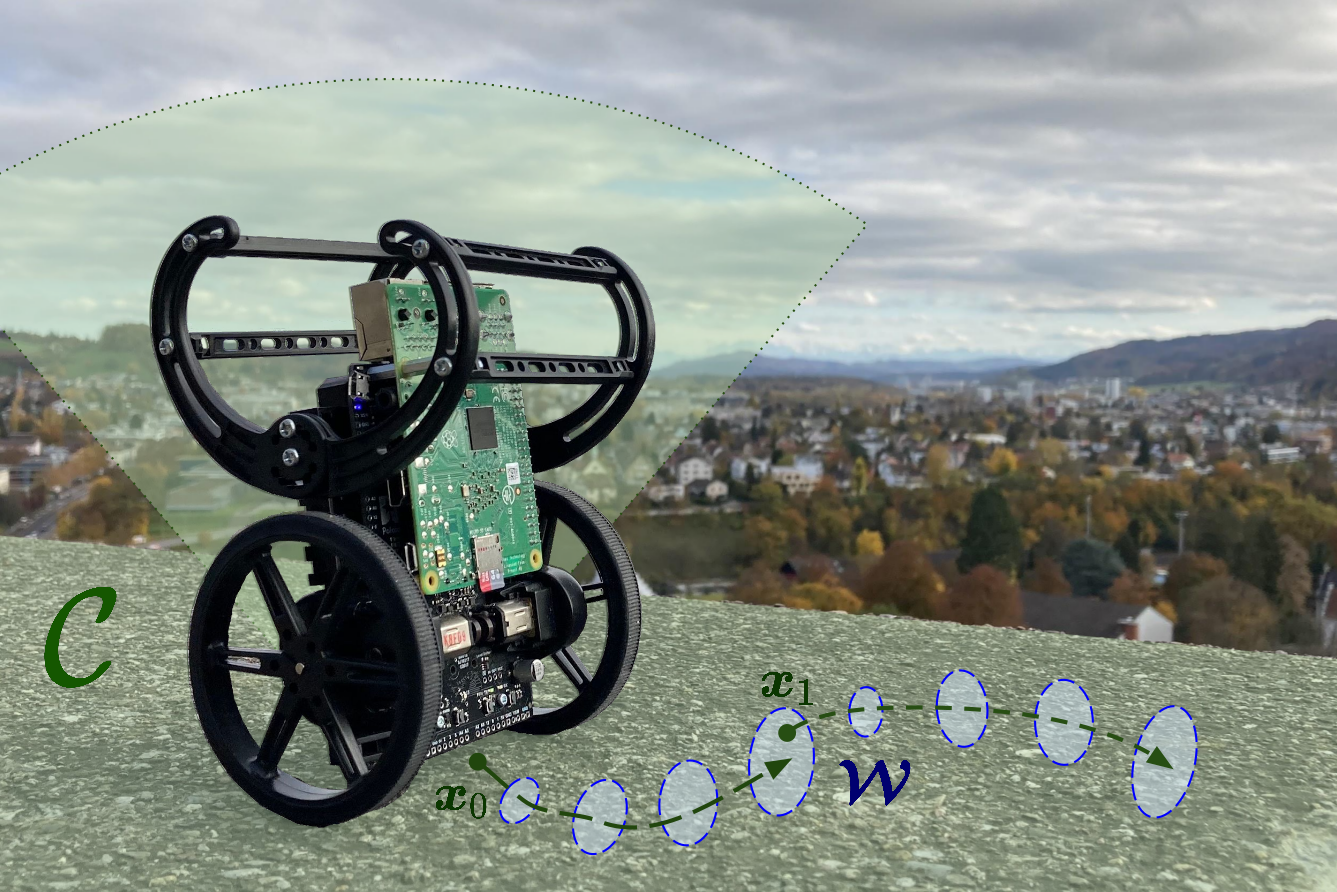}
	\caption{Safety-critical systems require fast and tractable control algorithms. These should stabilize the system, while always remaining safe. For instance, the mini-Segway shown in the figure must follow a reference position without falling down, i.e., satisfy safety bounds on the pitch angle at all times.}
	\label{fig:mini-segway_1}
\end{figure}

A popular tactic in ensuring both safety and stability is to de-couple the two tasks \cite{magdici2017adaptive,schouwenaars2006safe,bak2009system,gurriet2018towards}. To this end, a stabilizing controller can be applied jointly with some form of safety certificate function, such as a Control Barrier Function (CBF). Specifically, for nonlinear control-affine systems, CBFs can be used to reformulate nonlinear, nonconvex constraints as affine constraints point-wise in time. CBFs have been applied to discrete \cite{agrawal2017discrete}, stochastic \cite{ahmadi2019safe,takano2018application}, and hybrid systems \cite{nguyen20163d} as a part of control Lyapunov function-based controllers \cite{ames2016control,xu2015robustness,nguyen20163d}, 
and dedicated safety filters \cite{gurriet2018towards}. However, the notion of safety in CBFs relies on the assumption that one has continuous access to the system states and has the ability to continuously modify the states. However, these assumptions do not hold in systems with discrete-time control schemes \cite{gurriet2019realizable}.
An extension of CBFs to such systems is presented in \cite{singletary2020control}, where the point-wise CBF constraints must be robustly satisfied over a set. This leads to robust optimization problems, which may be difficult to solve.

\paragraph*{Contributions} This paper proposes a new tractable approach to safely control continuous-time dynamical systems using discrete-time controllers. Specifically:
\begin{itemize}
    \item We extend CBFs to account for discretization error, and guarantee constraint satisfaction for nonlinear control-affine systems under \textit{discrete-time} nominal controllers.
    \item We combine CBFs and approximate NMPC into an efficient algorithm with continuous-time safety guarantees.
    \item We validate our proposed approach in hardware experiments on a mini-Segway, and demonstrate the need to account for discretization errors to guarantee safety at all times given limited computational resources.
\end{itemize}    

\paragraph*{Notation} For a vector $v \in \R^n$ or vector-valued function $v : \R^n \to \R^m$, we use $v_i$ to denote its $i$-th component. We denote the Euclidean norm as $\norm{v} = \sqrt{\tran{v}v}$. A function ${f : \R^n \to \R^m}$ is Lipschitz continuous on a set $\mathcal{X} \subseteq \R^n$ if there exists $L_f \in \R$ such that $\norm{f(x) - f(y)} \leq L_f\norm{x - y}$ for all $x,y\in\mathcal{X}$. If $f$ is also differentiable, then $\norm{\grad f(x)} \leq L_f$ for all $x \in \mathcal{X}$, so we use $L_f \defn \max_{x \in \mathcal{X}} \norm{\grad f(x)}$ when $\mathcal{X}$ is compact. We denote the Minkowski sum of two sets $\mathcal{X} \subset \R^n$ and $\mathcal{Y} \subset \R^n$ as $\mathcal{X} \oplus \mathcal{Y}$, and similarly the Pontryagin difference as $\mathcal{X} \ominus \mathcal{Y}$. We denote the interior and boundary of a set $\mathcal{S}$ as $\mathrm{int}(\mathcal{S})$ and $\partial \mathcal{S}$, respectively.




\section{Continuous Time Control Barrier Functions}\label{sec:conttimecbf}

Consider the nonlinear control-affine dynamical system
\begin{equation}\label{eq:controlaffinesystem}
	\dot{x}(t)	= f(x(t)) + B(x(t))u(t),
\end{equation}
with state $x(t) \in \R^n$ and control input $u(t) \in \R^m$, where $f: \R^n \rightarrow \R^n$ and $B: \R^n \rightarrow \R^{n\times m}$ are Lipschitz continuous. In this section, we largely follow prior work \cite{wieland2007constructive,ames2016control,ames2019control} to formalize ``safety'' as controlling the state of \eqref{eq:controlaffinesystem} to remain within a designated safe set $\mathcal{X} \subset \R^n$ using only control inputs from an admissible set $\mathcal{U} \subset \R^m$. To this end, we search for a subset $\mathcal{C} \subseteq \mathcal{X}$ which is \emph{controlled invariant}, i.e., such that for each $x(0) \in \mathcal{C}$ there exists an admissible input trajectory $u(t) \in \mathcal{U}$ such that $x(t) \in \mathcal{C}$ for all $t \geq 0$. Our objective is to render this $\mathcal{C}$ \textit{forward invariant}, i.e. design a controller, such that $x(t) \in \mathcal{C}$ for all $t \geq 0$. Specifically, we restrict ourselves to the case where, given a continuously differentiable function $h : \mathcal{X} \to \R$ such that $\grad{h}(x) \neq 0$ whenever $h(x) = 0$, we define $\mathcal{C}$ as the super-level set
\begin{equation}\label{def:safeset}
	\mathcal{C} = \{x \in \mathcal{X} \mid h(x) \geq 0\},
\end{equation}
and assume $0\in\mathcal{C}$ without loss of generality. 
The description in \eqref{def:safeset} conveniently ensures that $\mathcal{C}$ is controlled invariant if and only if for each $x \in \partial{C}$, there exists an input $u \in \mathcal{U}$ such that
\begin{equation}\label{eq:alternativerepresentationofcontrolledinvariance}
	\dot{h}(x,u) = \tran{\grad{h}(x)}(f(x) + B(x)u) \geq 0.
\end{equation}
We use property \eqref{eq:alternativerepresentationofcontrolledinvariance} to define CBFs.

\begin{definition}[Control Barrier Function]
	Let $h : \mathcal{X} \to \R$ be continuously differentiable and satisfy $\grad{h}(x) \neq 0$ whenever $h(x) = 0$. Then $h$ is a Control Barrier Function (CBF) for the dynamical system \eqref{eq:controlaffinesystem} if there exists an extended class-$\mathcal{K}$ function\footnote{A continuous function $\beta : (-b,a) \rightarrow (-\infty,\infty)$ for some $a,b > 0$ is said to belong to extended class-$\mathcal{K}$ if $\beta$ is strictly increasing and $\beta(0) = 0$.} $\alpha : \R \to \R$ such that
	\begin{equation}\label{eq:CBFsupFormulation}
		\sup_{u \in \mathcal{U}} \tran{\grad{h}(x)}(f(x) + B(x)u) \geq -\alpha(h(x)),
	\end{equation}
	for all $x$ satisfying $h(x) \geq 0$.
\end{definition}

Finding a CBF for the system~\eqref{eq:controlaffinesystem} is sufficient to guarantee that~\eqref{eq:alternativerepresentationofcontrolledinvariance} holds and hence that the system is safe. We slightly rearrange~\eqref{eq:CBFsupFormulation} into
\begin{equation}\label{eq:CBFFormulation}
	\mathrm{CBF}(x,u) \defn \tran{\grad{h}(x)}(f(x) + B(x)u) + \alpha(h(x)) \geq 0,
\end{equation}
which we term the \emph{affine CBF condition} to highlight that it is indeed affine in the control input~$u$. Thus, \eqref{eq:CBFFormulation} can be embedded as a simple affine constraint to design safe optimization-based controllers \cite{ames2016control} and safety filters \cite{gurriet2018towards}. 

However, the affine CBF condition \eqref{eq:CBFFormulation} assumes the control signal~$u(t)$ is applied in continuous-time, while in practice, controllers operate at discrete time instants. Thus, we generally lose the safety guarantees provided by CBFs. Moreover, prior work on CBFs implicitly assumes the control input $u$ chosen to satisfy the affine CBF condition~\eqref{eq:alternativerepresentationofcontrolledinvariance} lies in the admissible set~$\mathcal{U}$, while in general this may not hold.

\section{Problem Definition}

The objective of this work is to design a control law using only admissible inputs from $\mathcal{U}$ which steer the system~\eqref{eq:controlaffinesystem} to a desired state $x_d \in \mathcal{X}$ \emph{safely}, i.e., we require $x(t) \in \mathcal{X}$ for all $t \geq 0$. Specifically, we assume $\mathcal{X} \subset \R^n$ is a compact set encoding any safety constraints, while $\mathcal{U} \subset \R^m$ is a convex polytopic set encoding control input constraints. Furthermore, we consider discrete-time controllers of the form 
\begin{equation}\label{eq:dt_controller}
	u(\tau) = u_k,\ \forall \tau \in [t_k, t_{k+1}).
\end{equation}
The input $u_k \in \mathcal{U}$ is computed at the time instant ${t_k \geq 0}$ and applied over the interval $[t_k, t_{k+1})$, where ${t_{k+1} - t_k = T}$ for all $k \in \mathbb{N}_{0}$ and some sampling time $T > 0$, which corresponds to a zero-order hold.

As stated in Section~\ref{sec:conttimecbf}, controllers derived using CBFs rely on the continuous enforcement of the affine CBF condition~\eqref{eq:CBFFormulation}, while in practice this condition can only be enforced at each sample time $t_k$. This discretization contradicts prior continuous-time analyses, thus safety constraints may not hold at all times $\tau \in [t_k,t_{k+1})$. Alternatively, NMPC controllers only consider constraints enforcement at a finite number of discretization nodes $\{t_k\}_{k=0}^N$. Although both approaches may lead to controllers satisfying constraints at all times given \emph{fast enough} update rates, computational limitations motivate a finer analysis to explicitly account for this type of error and guarantee constraint satisfaction at all times.


\section{Discrete-Time Control Barrier Functions with Continuous Constraint Satisfaction}\label{controlledinvariancefordtcontrollers}

In this section, we derive two controllers of the form \eqref{eq:dt_controller} which guarantee constraint satisfaction at all times. To this end, we first characterize the maximum variation of the CBF condition \eqref{eq:CBFFormulation} over a time interval, and use it to derive a discrete-time CBF condition which results in a general and intuitive controller which explicitly accounts for the discretization error. For the second controller we propose, we exploit additional problem structure in the form of an auxiliary controller. We then characterize the difference between an ideal but intractable continuous-time controller which leverages the affine CBF condition \eqref{eq:CBFFormulation}, and a tractable discrete-time approximation which follows \eqref{eq:dt_controller}. With this analysis, we provide a tube-based CBF controller with continuous-time safety guarantees. 

\subsection{Discrete Barrier Condition}\label{subsec:discretebarriercondition}
Consider the affine CBF condition \eqref{eq:CBFFormulation} for some safe set~$\mathcal{C}$ and a corresponding CBF~$h$. For a discrete-time controller~\eqref{eq:dt_controller}, we want to bound the maximum change in the CBF condition over some time to construct a guarantee that holds over an entire time interval. To this end, consider some differentiable, Lipschitz continuous function ${\phi : \R^n \to \R^d}$. To bound the change in $\phi$ along the trajectory $(x(t), u(t))$, we seek a constant $\widetilde{L}_\phi$ such that
\begin{equation}
    \norm{\phi(x(t + T)) - \phi(x(t))} \leq \widetilde{L}_\phi T.
\end{equation}
By the chain rule, $\dot{\phi}(t) = \tran{\grad\phi(x(t))}(f(x(t)) + B(x(t))u(t))$, and thus we can use
\begin{equation}\label{eq:lipschitz_const_time}
    \widetilde{L}_\phi \defn L_\phi\max_{x\in\mathcal{X},u\in\mathcal{U}}\rbr{f(x) + B(x)u}.
\end{equation}
This augments the Lipschitz constant $L_\phi$ for $\phi$ into a Lipschitz-like constant $\widetilde{L}_\phi$ with respect to time along \emph{any} state-input trajectory of the dynamical system. Then, over the time interval $\tau \in [t_k, t_k + T)$, we can bound the evolution of~$\phi$ along $(x(t),u(t))$ according to
\begin{equation}\label{eq:boundedevolution}
    \phi(x(\tau)) \in \phi(x(t_k)) \oplus \mathcal{W}_\phi,\ \forall \tau \in [t_k, t_k + T),
\end{equation}
where we define the truncation error hypercube
\begin{equation}
    \mathcal{W}_\phi \defn \{w \in \R^d \mid \norm{w}_\infty \leq \widetilde{L}_\phi T\}.
\end{equation}
The affine CBF condition~\eqref{eq:CBFFormulation} depends on $f$, $B$, $h$, and $\grad{h}$, so we can treat
\begin{equation}\label{eq:disturbanceset}
    \mathcal{W} \defn \mathcal{W}_f \x \mathcal{W}_B \x \mathcal{W}_h \x \mathcal{W}_{\grad h}
\end{equation}
as a set of possible disturbances $w = (w_f, w_B, w_h, w_{\grad h})$ which should be accounted for to guarantee safety at all times $\tau \in [t_k,t_{k+1})$. From this observation, we introduce the \emph{Discrete-time Barrier Condition (DBC)}
\begin{equation}\begin{aligned}
    &\mathrm{DBC}(x,u,w) \\
    &\defn \tran{(\grad{h}(x) + w_{\grad{h}})}\!\rbr{f(x) + w_f + (B(x) + w_B)u} \\
    &\quad + \alpha(h(x) + w_h).
\end{aligned}\end{equation}
At time $t_k$, this discrete barrier condition captures all possible system evolutions, thus guaranteeing that the affine CBF condition \eqref{eq:CBFFormulation} holds for an entire time interval $\tau \in [t_k,t_{k+1})$. Further, if a discrete time controller $u_k$ is synthesized such that it enforces this condition at all discrete times $t_k$, then the system is safe for all times under this controller. These two statements are formalized in the following theorem.
\begin{theorem}[Controlled Invariance Using a DBC]\label{thm:discreteCBF}
	Consider the dynamical system in \eqref{eq:controlaffinesystem} and a safe set $\mathcal{C}$ with a corresponding CBF $h$. Assume that the control law $u(\tau) = u_k \in \mathcal{U}$, $\tau \in [t_k,t_{k+1})$ satisfies 
	\begin{equation}\label{eq:DCBF}
	    \mathrm{DBC}(x(t_k),u_k,w) \geq 0,
	\end{equation}
	for all $w \in \mathcal{W}$, where $\mathcal{W}$ is defined as in \eqref{eq:disturbanceset}. Then $\mathcal{C}$ is forward invariant for \eqref{eq:controlaffinesystem} at all times $\tau \in [t_k,t_{k+1})$. Furthermore, if $x(0) \in \mathcal{C}$ and the discrete-time controller with $u_k \in \mathcal{U}$ satisfies \eqref{eq:DCBF} at every time step $t_k$ for $k \in \mathbb{N}_0$, then the system is safe for all $t \geq 0$.
\end{theorem}
\begin{proof} 
    Denote $x_t \defn x(t)$ for conciseness. We start with the proof of the first claim, and show that \eqref{eq:DCBF} holding at time~$t_k$ implies that \eqref{eq:CBFFormulation} holds for all times $\tau \in [t_k,t_{k+1})$. According to~\eqref{eq:boundedevolution}, let $w \in \mathcal{W}$ be such that ${f(x_\tau) = f(x_{t_k}) + w_f}$, $B(x_\tau) = B(x_{t_k}) + w_B$, $h(x_\tau) = h(x_{t_k}) + w_h$, and $\grad{h}(x_\tau) = \grad{h}(x_{t_k}) + w_{\grad{h}}$. Substituting these into the affine CBF condition~\eqref{eq:CBFFormulation} along with the given discrete-time control input yields
    \begin{equation*}\begin{aligned}
        &\mathrm{CBF}(x_\tau,u_k) \\
        &= \tran{\grad{h}(x_\tau)}(f(x_\tau) + B(x_\tau)u_k) + \alpha(h(x_\tau)) \\
        &= \tran{(\grad h(x_{t_k}) + w_{\grad h})}(f(x_{t_k}) + w_f + (B(x_{t_k}) + w_B)u_k) \\
        &\quad + \alpha(h(x_{t_k}) + w_h) \\
        &= \mathrm{DBC}(x_{t_k}, u_k, w)
    \end{aligned}\end{equation*}
    If $\mathrm{DBC}(x_{t_k}, u_k, w) \geq 0$ for all $w \in \mathcal{W}$, then this holds for the particular $w$ above, so~\eqref{eq:DCBF} ensures $\mathrm{CBF}(x_{\tau},u_k) \geq 0$ for all $\tau \in [t_k,t_{k+1})$. This, and the fact that~$h$ is a CBF, guarantee $x(\tau)\in\mathcal{C}$ for all $\tau \in [t_k,t_{k+1})$ as long as $x(t_k)\in\mathcal{C}$.
    
    The second statement follows; by assumption, $x(t_k) \in \mathcal{C}$ and there exists a $u_k \in \mathcal{U}$ such that \eqref{eq:DCBF} holds. By induction with the previous result, $\mathcal{C}$ is forward invariant under this discrete-time controller.
\end{proof}

\begin{remark} 
    We can reduce conservatism of our derivations in two ways. First, we can define the Lipschitz constant component-wise as $\|f(x)-f(x_0)\|\leq \smash{\sum_{i=1}^n} L_{f,i} |x_i-x_{0,i}|$, which is equivalent to normalizing each dimension of the dynamics. Second, local upper bounds can reduce conservatism of the Lipschitz-like constants.
\end{remark}

Since~\eqref{eq:DCBF} is nonconvex, it can be challenging to enforce for general dynamical systems. Inspired by work on CBF-based robust controllers \cite{gurriet2018towards}, we relax~\eqref{eq:DCBF} to a set of affine conditions, which can then be used to construct safety filters and safe controllers. To this end, we write~\eqref{eq:DCBF} as
\begin{equation}
    \mathrm{DBC}(x,u,w) = -\tran{a(x,w)}u - b(x,w),
\end{equation}
where
\begin{equation*}\begin{aligned}
    a(x,w) &\defn -\tran{(B(x) + w_B)}(\grad{h}(x) + w_{\grad{h}}) \\
    b(x,w) &\defn - (\grad{h}(x) + w_{\grad{h}})(f(x) + w_f) - \alpha(h(x) + w_h).
\end{aligned}\end{equation*}
Then $\mathrm{DBC}(x,u,w) \geq 0$ for all $w \in \mathcal{W}$ if and only if
\begin{equation}\label{eq:nonconvex_cbf_opt_condition}
    \max_{w \in \mathcal{W}} \rbr{\tran{a(x,w)}u + b(x,w)} \leq 0.
\end{equation}
For fixed $u$, the left-hand side of \eqref{eq:nonconvex_cbf_opt_condition} is nonconvex in~$w$. A sufficient relaxed condition for~\eqref{eq:nonconvex_cbf_opt_condition} to hold is
\begin{gather}
\label{eq:inproof}
\max_{\tilde{a}_{j}, \tilde{b}}  \ 
\Big(\tilde{a}^\textsf{T}u + \tilde{b}\Big) \leq 0,
\\
\text{s.t. } \ 
\tilde{a}_{j} \in\mathcal{A}^j_{\mathcal{W}} = \{a_{j}(x,w) | w \in \mathcal{W}\}, \ 
\text{ } \tilde{b} \in\{b(x,w) | w \in \mathcal{W}\}, 
\nonumber
\end{gather}
where we omit the dependency on $x$ for conciseness, and 
$j=1,\dots,m$. 
$\mathcal{A}^j_{\mathcal{W}}$ are nonconvex sets. 
Compared to \eqref{eq:nonconvex_cbf_opt_condition}, this last condition is conservative, since  a different disturbance $w$ can be chosen for each $j$-th dimension of $a$. 
As $\mathcal{W}$ and $a$ can be nonconvex, \eqref{eq:inproof} is nonconvex. 
This robust constraint can be relaxed as a set of affine constraints.

However, they are compact and thus can be outerbounded conservatively as polytopic sets. 
In this work for tractability, we consider conservative hyperrectangular sets of the form
\begin{align*}
\mathcal{D}^j_{\mathcal{W}} &= 
\{
\tilde{a}_j
\ | \ 
\tilde{a}_j \leq \bar{a}_j, 
\,
-\tilde{a}_j \leq -\underline{a}_j
\},\quad j=1,\dots,m,
\\
\mathcal{B}_{\mathcal{W}} &= 
\{
\tilde{b}
\ | \ 
\tilde{b} \leq \bar{b}, 
\,
-\tilde{b} \leq -\underline{b}
\},
\end{align*}
where 
$\underline{a}_j = \min_{\tilde{a}_j\in\mathcal{A}^j_{\mathcal{W}}} \tilde{a}_j$, and 
$\bar{a}_j = \max_{\tilde{a}_j\in\mathcal{A}^j_{\mathcal{W}}} \tilde{a}_j$. 
With these rectangular sets, we can rewrite the condition in \eqref{eq:inproof} as affine constraints of the form $D[\tilde{a},\tilde{b}]^\textsf{T}\leq d(x)$, where
\begin{equation}
D^\textsf{T}
{=}\, 
\AverageSmallMatrix{
1 & -1   & \ \mydots & \ 0 & \ 0 \\
0 &  \ 0 & \ \mydots & \ 0 & \ 0 \\
0 &  \ 0 & \ \mydots & \ 1 &  -1
},
\ 
d(x)^\textsf{T} 
{=}\, 
\AverageSmallMatrix{
\bar{a}_1 & 
-\underline{a}_1 & 
\mydots & 
\bar{a}_m & 
-\underline{a}_m & 
\bar{b} & 
-\underline{b}
}.
\end{equation}
Then, a sufficient condition for \eqref{eq:inproof} is given as
\begin{align}
\label{eq:test2}
\max_{\tilde{a}, \tilde{b}}  
\left(\tilde{a}^\textsf{T}u + \tilde{b}\right) \leq 0, 
\quad
\text{s.t. } \ 
D[\tilde{a},\tilde{b}]^\textsf{T}\leq d(x).
\end{align}
\eqref{eq:test2} is a linear program in $\tilde{a}$ and $\tilde{b}$ (guaranteed feasible by a proper choice of $\mathcal{C}$). Denoting  
$\tilde{u} \triangleq [u, 1]^\textsf{T}\in\R^{m+1}$ 
for the concatenation of $u$ with the scalar $1$, we can express its dual as
\begin{align}
\label{eq:test3}
\max_{\tilde{\lambda}}  &\left( d(x)^\textsf{T}\tilde{\lambda}\right) \leq 0, 
\quad
\text{s.t. } \ 
D^\textsf{T}\tilde{\lambda}= \tilde{u}, 
\quad 
\tilde{\lambda}\geq 0 ,
\end{align}
where $\tilde{\lambda}\in \R^{2(m+1)}$ are the dual variables. 
Since \eqref{eq:test2} is convex, its dual problem \eqref{eq:test3} is equivalent by strong duality. 
Any feasible solution of this problem will guarantee safety. These solutions must necessarily satisfy the following set of affine conditions
\begin{align}\label{eq:affineformulationthmone}
d(x)^\textsf{T}\tilde{\lambda}\leq 0,  \ \ \
D^\textsf{T}\tilde{\lambda}= \tilde{u}, \ \ \
\tilde{\lambda}\geq 0. 
\end{align}
Using these affine constraints, we propose an optimization-based controller which guarantees safety at all times. Specifically, given a nominal control input $u_k$, it consists of solving, at each time $t_k$, the following quadratic program (QP):
\begin{align}
\label{eq:affconditionsafetyfilter}
\min_{u, \tilde{\lambda}} \ \ &\| u - u_k \|^2, 
\quad \textrm{s.t.} \ \ \eqref{eq:affineformulationthmone}.
\end{align}

To summarize, the affine conditions \eqref{eq:affineformulationthmone} allow to continuously guarantee safety for our system \eqref{eq:controlaffinesystem}, if satisfied at discrete times $t_k$ only. The resulting formulation can be used as a safety filter or to synthesize controllers in combination with optimization-based control approaches, such as Control Lyapunov Functions-based controllers or MPC. We discuss the advantages and limitations in section \ref{sec:results}.

\subsection{Tube-CBF}
\label{subsec:tubecbf}
We present our second algorithm, Tube-CBF, which in contrast to the DBC, we uses the unaltered affine CBF condition from section \ref{sec:conttimecbf}.

First, consider a nominal case where we have an \emph{ideal} continuous-time controller fulfilling the affine CBF condition at all times and denote the nominal trajectory as $\bar{x}(t)$. As we use a discrete-time controller \eqref{eq:dt_controller}, naturally, discretization leads to a difference between the resulting trajectory $x(t)$ and $\bar{x}(t)$. As the affine CBF condition only guarantees safety for the nominal trajectory $\bar{x}(t)$, the true trajectory $x(t)$ is potentially outside of the safe set and may be unsafe. Inspired by tube MPC, we propose the use of a discrete-time auxiliary controller $\kappa$ to regulate the error $z(t) = x(t) - \bar{x}(t)$ to zero and ideally bound it in an invariant tube. Knowledge of such an auxiliary controller and a corresponding invariant tube together with the affine CBF condition, allows us to design a discrete-time controller \eqref{eq:dt_controller} with continuous-time safety guarantees. To this end, we first define a robust invariant set.

\begin{definition}[Robust Invariant Set with Discrete Feedback]\label{def:robust_inv_set_error}
	A set $\Omega \subset \mathcal{C} \subset \R^{n}$ is a robust control invariant set for the error $x-\bar{x}$ 
	if there exists an auxiliary feedback control law $\kappa:\R^n\times\R^n\rightarrow\R^m$ of the form \eqref{eq:dt_controller},  
	such that 
	$\kappa(x,\bar{x}) \in \mathcal{U}$ for all $x,\bar{x}\in\mathcal{C}$, 
	and under this controller, 
	if $x(0)-\bar{x}(0) \in \Omega$, 
	then $x(t)-\bar{x}(t) \in \Omega$, for all $t \geq 0$.
\end{definition}

Note that finding $\Omega$ and $\kappa$ is in general not straightforward. A rather coarse over-approximation of $\Omega$ with a corresponding $\kappa$ can be found in \cite{yu2013tube}.
This leads to the following assumption, which is strong but common in tube NMPC literature \cite{yu2013tube}:

\vspace{1mm}
\begin{assumption} 
	\label{ass:omega}
	Suppose we have a control law $\kappa$ of the form \eqref{eq:dt_controller}, and a set $\Omega$,  such that $\Omega$ is robust control invariant, i.e., $(\kappa,\Omega)$ satisfy definition \ref{def:robust_inv_set_error}.
\end{assumption}

We use the auxiliary controller to compensate discretization errors in addition to the nominal controller. In the presence of control input constraints, we need to account for the contribution to the control input from the auxiliary controller. We capture this in the following set:
$$G := \text{Poly}\left(\{
\kappa(x,\bar{x}) 
\,|\, 
\bar{x} \in \mathcal{C} \ominus \Omega, 
\ \text{and} \ 
x-\bar{x} \in \Omega
\}\right) \subset \R^m,
$$ 
where $\text{Poly}()$ denotes the smallest polytopic hull. To account for the error $z \in \Omega$, we define a reduced compact safe set $\mathcal{C}'$ as summarized in the following definition:

\begin{definition}[Reduced Safe Set]\label{def:reducedcompactsafeset}
	A reduced safe set $\mathcal{C}'$ to a set $\mathcal{C}$ is a compact set such that 
	$\mathcal{C}' \subseteq \mathcal{C} \ominus \Omega \subset \R^{n}$, where $\mathcal{C}' $ is characterized by a CBF $h$ under $\mathcal{U}' := \mathcal{U} \ominus G$.
\end{definition}

\begin{assumption}\label{ass:reducedcompactsafeset}
	We have access to a  reduced safe set $\mathcal{C}'$ according to definition \ref{def:reducedcompactsafeset}.
\end{assumption}

Building on assumptions \ref{ass:omega} and \ref{ass:reducedcompactsafeset}, we propose the Tube-CBF algorithm.

\begin{figure}[t]
	\centering
	\includegraphics[width=0.95\linewidth]{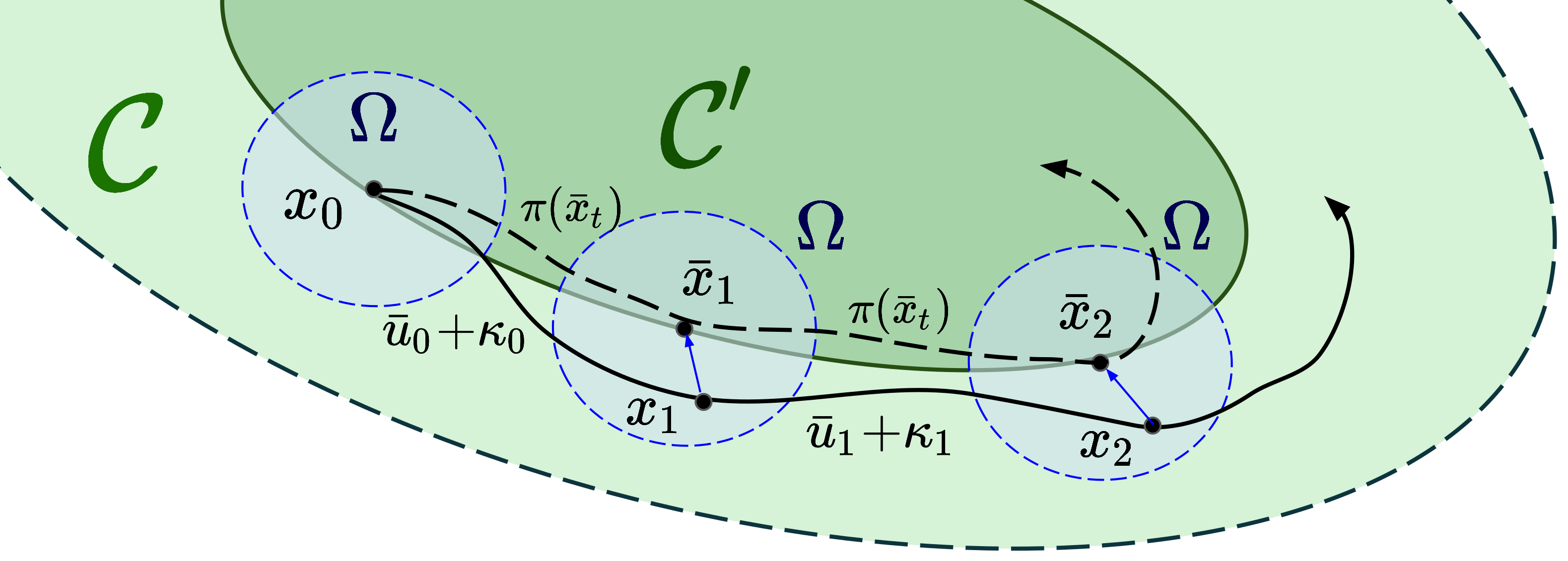}
	\vspace{-2mm}
	\caption{The Tube-CBF algorithm consists of applying the affine CBF condition to a reduced set $\mathcal{C}'$. 
	Then, the discretization error arising between the trajectories under an ideal continuous time controller (denoted as $\bar{x}$), and a discrete counterpart (denoted as $x$) is kept in the invariant tube $\Omega$ with an auxiliary controller $\kappa$.}
	\label{fig:theoremtwo}
	\vspace{-5mm}
\end{figure}

\textbf{Tube-CBF algorithm:}
\begin{enumerate}
	\label{alg:tubenmpcCBF}
		\item At time $t_0$, set $\bar{x}(t_0) = x(t_0) \in \mathcal{C}'$.
		\item At time $t_k$, find a nominal input $\bar{u}(t_k) \in \mathcal{U}'$, such that the affine CBF condition for $\mathcal{C}'$ holds at $\bar{x}(t_k)$. 
		\item Apply the control input $u(t_k) = \bar{u}(t_k) + \kappa(x(t_k), \bar{x}(t_k))$ to the system during the time interval $\tau \in [t_k,t_{k+1})$. 
		\item At time $t_{k+1}$, measure the state $x(t_{k+1})$ and find a state $\bar{x}(t_{k+1}) \in \mathcal{C}'$ such that $x(t_{k+1}) \in \bar{x}(t_{k+1}) \oplus \Omega$ and go to step 2.
\end{enumerate}
Figure \ref{fig:theoremtwo} depicts multiple time steps of the algorithm. Furthermore, we formalize the safety guarantees in the following theorem.

\begin{theorem}[Controlled Invariance Using Tube-CBF]\label{thm:tubeCBF}
	Consider the dynamical system in \eqref{eq:controlaffinesystem} and a safe set $\mathcal{C}$. Suppose  assumptions \ref{ass:omega} and \ref{ass:reducedcompactsafeset} hold.  Then, $\mathcal{C}$ is forward invariant under the Tube-CBF algorithm for \eqref{eq:controlaffinesystem} under \eqref{eq:dt_controller} at all times $t\geq 0$ if $x(0)\in\mathcal{C}'$. 
\end{theorem}



\vspace{1mm}

\begin{proof}
We proceed by induction and consider first the induction step.
At time $t_k$, we assume we have $\bar{x}(t_k)\in\mathcal{C}'$ and $x(t_k) \in \bar{x}(t_k) \oplus \Omega$. Then, according to Step 2 of the Tube-CBF algorithm, we compute an input $\bar{u}(t_k)$ that satisfies the CBF condition for $\mathcal{C}'$ at $\bar{x}(t_k)$. A feasible control input exists by the definition of $\mathcal{C}'$. Next we apply the control input $u(t_k)=\bar{u}(t_k)+\kappa(x(t_k), \bar{x}(t_k))$ to the system over a time interval $t\in [t_k,t_{k+1})$ according to Step 3. 
By definition of $(\kappa,\Omega)$ and $\mathcal{C}'$, $u(t_k)\in\mathcal{U}$. 
Further, by definition of $(\kappa,\Omega)$, 
this guarantees that $x(\tau)-\bar{x}(\tau) \in \Omega$ for $\tau \in [t_k,t_{k+1})$. 
Also, $\bar{x}$ follows an ideal controller, such that $\bar{x}(\tau)\in\mathcal{C}'$ for $\tau \in [t_k,t_{k+1})$. 
Therefore, at Step 4, there exists a $\bar{x}(t_{k+1})$ such that
$x(t_{k+1}) \in \bar{x}(t_{k+1})\oplus \Omega$. We continue with $\bar{x}(t_{k+1})$ at step 2.

The initial condition $x(t_0) = \bar{x}(t_0) \in \mathcal{C}'$ completes the proof.

\end{proof}



\section{Continuous Safety Guarantees with NMPC}
\label{sec:NMPC}

In this section, we design an optimal controller of the form \eqref{eq:dt_controller}. We consider a direct NMPC approach to reduce the optimal control problem to a nonlinear program (NLP). We discretize the system dynamics \eqref{eq:controlaffinesystem}
\begin{equation}
	\label{eq:discretization}
	x(t_{k+1}) = f_k^d(x_k,u_k) = x(t_k) +\hspace{-0.1cm} \int_{t_k}^{t_{k+1}} \hspace*{-0.6cm}[f(x(\tau)) + B(x(\tau))u_k] \mathrm{d}\tau,
\end{equation}
where the integral in \eqref{eq:discretization} can be approximated using a numerical integration scheme of choice. In NMPC, we solve the NLP at each time step $t_k \geq 0$ to obtain the controller \eqref{eq:dt_controller}. The NLP formulation we consider is

\begin{align}
	\label{eq:NLP}
	\min_{X,U,S} l_N(x_N) + \zeta_N(s_{2N-1}) &+ \sum_{i=0}^{N-1} l_k(x_i,u_i) + \zeta_i(s_{2i},s_{2i+1}) \nonumber \\
	\text{s.t.} 
	\quad 
	x_{i+1} - f_i^d(x_i,u_i) &= s_{2i}, \quad i = 0, \ldots, N-1 \\
	g_i(x_i,u_i) &\leq s_{2i+1}, \ i = 0, \ldots, N-1 \nonumber 
	\\
	\qquad 
	\qquad
	\qquad
	x_0 - \hat{x} &= 0, \quad 
	g_N(x_N) \leq s_{2N-1},  
	\nonumber
\end{align}
with $X = [x_0, \ldots, x_N]^\textsf{T}$, $U = [u_0, \ldots, u_{N-1}]^\textsf{T}$, and $S = [s_0, \ldots, s_{2N-1}]^\textsf{T}$. At time $t_k$, the NLP is initialized with $\hat{x}$. Along the prediction horizon $i = [0,\ldots,N]$, the functions $l_k$ and $l_N$ denote the stage and terminal cost, while the functions $g_k$ and $g_N$ denote (in)-equality constraints. To guarantee feasibility of the NLP, we soften the constraints and introduce the slack variables $S$, which we penalize with a function of the form $\zeta_i(s_i) = \| s_i \| + s_i^2$.  To facilitate reading, we collect all decision variables in $v = [X^\textsf{T}, U^\textsf{T}, S^\textsf{T}]^\textsf{T}$ and split the constraints into equality constraints $I$ and inequality constraints $G$ to arrive at a general NLP formulation with the cost function $J$,
\begin{equation}
	\label{eq:generalNLP}
	\min_v J(v) \quad \text{s.t.} \ I(v) = 0, \ G(v)  \leq 0.
\end{equation}
We tackle \eqref{eq:generalNLP} using Sequential Quadratic Programming (SQP) \cite{nocedal2006sequential}. In SQP, we solve a series of  approximated subproblems to iteratively arrive at a solution to \eqref{eq:generalNLP}. Specifically, SQP methods minimize a second-order Taylor expansion of the Lagrangian of \eqref{eq:generalNLP},
\begin{equation}
	\
	\Lagr(v,\lambda,\mu) = J(v) + \lambda^\textsf{T}I(v) + \mu^\textsf{T}G(v),
\end{equation}
where $\lambda$ and $\mu \geq 0$ are the Lagrange multipliers. We start out with an initial guess $v^0$ and incrementally update it with computed update steps $\delta w^i$ according to $v^{i+1} = v^i + \delta v^i$. Minimizing  the second-order Taylor expansion $\Lagr_{SQP}(\delta v^i, \delta \lambda^i, \delta\mu^i) = T_{\Lagr(v^i,\lambda^i,\mu^i)}(\delta v^i, \delta \lambda^i, \delta\mu^i) $ is equivalent to solving the following (potentially nonconvex) QP
\begin{align}
	\label{eq:SQP}
	\min_{\delta v^i} \qquad &\nabla_v J(v^i)^\textsf{T} \delta v^i + \frac{1}{2} {\delta v^i}^\textsf{T} H(v^i,\lambda^i,\mu^i)\delta v^i \\
	\text{s.t.} \qquad &I(v^i) + \nabla_v I(v^i)^\textsf{T} \delta v^i = 0, \nonumber \\
	&G(v^i) + \nabla_v G(v^i)^\textsf{T} \delta v^i \leq 0, \nonumber
\end{align}
where $H = \nabla_v^2 \Lagr(v^i,\lambda^i,\mu^i)$ denotes the Hessian. The optimization variables are then updated according to
\begin{align}
	w^{i+1} = w^i + \delta w^i, \ \ 
	\lambda^{i+1} = \lambda^{i}_{QP}, \ \ 
	\mu^{i+1} = \mu^{i}_{QP}. \nonumber
\end{align}

Solving a sequence of QPs until convergence of the underlying NLP is still computationally intense. Thus, we opt for an approximate NMPC scheme, also known as the real-time iteration scheme (RTI). The RTI is designed for NMPC problems with a quadratic cost function and uses the Gauss-Newton approximation for the Hessian. Furthermore, in RTI, we only solve a single QP at each time step $t_k$, which leads to a non-converged and generally suboptimal solution. This non-converged solution poses a few challenges. First, the linear approximation of the constraints in \eqref{eq:SQP} no longer guarantees strict constraint satisfaction of the original problem \eqref{eq:generalNLP}. Second, recursive feasibility is lost and \eqref{eq:generalNLP} might become infeasible at some point, which must be avoided at all cost in safety critical systems.

Now, we use Tube-CBF introduced in \ref{controlledinvariancefordtcontrollers} to combine the enhanced performance of RTI with safety and recursive feasibility. To this end, we return again to the RTI formulation in \eqref{eq:SQP} and complement it with two hard constraints. First, we add the constraint $\bar{u}_0 \in \tilde{\mathcal{U}}$. To facilitate notation, we use an affine operator $\Xi_{u_0}$ with $u_0 \equiv \Xi_{u_0}v$ to extract $u_0$ from our optimization variables $v$. Second, we  add the affine CBF condition to our CBF for $\mathcal{C}'$, where we use $x_0 \equiv \Xi_{x_0}v$. This leads to

\begin{align}
	\label{eq:tubeCBFRTI}
	\min_{\delta v^i} \qquad &\nabla_v J(v^i)^\textsf{T} \delta v^i + \frac{1}{2} {\delta v^i}^\textsf{T} H(v^i, \lambda^i, \mu^i) \delta v^i \\
	s.t. \qquad &I(v^i) + \nabla_v I(v^i)^\textsf{T} \delta v^i = 0, \nonumber \\
	&\tilde{G}(v^i) + \nabla_v \tilde{G}(v^i)^\textsf{T} \delta v^i \leq 0, \nonumber \\
	&-\mathrm{CBF}\left(\Xi_{x_0}v^i,\Xi_{u_0}(v^i + \delta v^i)\right) \leq 0, \nonumber \\ 
	&\Xi_{u_0}(v^i + \delta v^i) \in \mathcal{U}'. \nonumber
\end{align}

The formulation \eqref{eq:tubeCBFRTI} provides point-wise constraint satisfaction, but does not yet guarantee recursive feasibility. To obtain recursive feasibility, we additionally perform the Tube-CBF algorithm from section \ref{subsec:tubecbf}. The combination of RTI and Tube-CBF is summarized in algorithm \ref{alg2} and the following theorem.

\begin{algorithm} 
	\caption{ RTI with Tube-CBF} 
	\label{alg2} 
	\begin{algorithmic}
		\State initialize: $\bar{x}(t_0) = x(t_0)$ with $x(t_0) \in \mathcal{C}'$
		\While{IsRunning()}
		\State $\bar{u}_k$ $\leftarrow$ obtain by solving \eqref{eq:tubeCBFRTI} at time $t_k$
		\State $u(t_k) \leftarrow \bar{u}(t_k)+ \kappa(x(t_k), \bar{x}(t_k))$
		\State apply $u(t_k)$ to plant
		\State wait()
		\State measure $x(t_{k+1})$ 
		\State find $\bar{x}(t_{k+1})$ s.t. $x(t_{k+1}) \in \bar{x}(t_{k+1}) \oplus \Omega$
		\EndWhile
	\end{algorithmic} 
\end{algorithm} 

\begin{theorem}[Safe RTI Using Tube-CBF]\label{thm:TubeCBFRTI}
	Consider the dynamical system in \eqref{eq:controlaffinesystem} under the RTI controller \eqref{eq:tubeCBFRTI} and a reduced safe set $\mathcal{C}'$ with a corresponding CBF $h$ according to definition \ref{def:reducedcompactsafeset}. We assume that $x(t_0) \in \mathcal{C}'$. Then $\mathcal{C}$ is forward invariant for \eqref{eq:controlaffinesystem} under \eqref{eq:tubeCBFRTI}  at all times $t \geq 0$ and hence the constraints are strictly satisfied under the Tube-CBF algorithm. Furthermore, since the constraints are satisfied at all times, the RTI is also recursively feasible.
\end{theorem}

\begin{proof}
	We leverage theorem \ref{thm:tubeCBF} for the proof. Theorem \ref{thm:tubeCBF} states that $\mathcal{C}$ is forward invariant for all times $t \geq 0$ if $x(0) \in \mathcal{C}'$ and there is a discrete-time controller that provides a nominal input $\bar{u}(t_k) \in \mathcal{U}'$ such that $\text{CBF}(\bar{x}(t_k),\bar{u}(t_k)) \geq 0$ for all $k \geq 0$.
	
	First, note that since $\mathcal{U}$ and $G$ are convex polytopes, $\mathcal{U}'$ is a convex polytope. Thus, $u \in \mathcal{U}' \Leftrightarrow Au \leq b$ for some $A$ and $b$, i.e., we can encode $\bar{u} \in \mathcal{U}' $ as a set of affine conditions. From the definition of CBFs we have that $ \mathcal{U}' \cap \{u \ | \ \text{CBF}(\bar{x},u) \geq 0 \} \neq \varnothing$ for all $\bar{x} \in \mathcal{C}'$. Thus, there exists a $\bar{u}(t_k)$ such that $A\bar{u}(t_k) \leq b$ and $\text{CBF}(\bar{x},\bar{u}) \geq 0$ and hence \eqref{eq:tubeCBFRTI} is feasible for all $\bar{x} \in \mathcal{C}'$, as all other constraints are soft constraints. By assumption $x(0) \in \mathcal{C}'$ and thus we can use  theorem \ref{thm:tubeCBF} to guarantee forward invariance of  $\mathcal{C}$ and correspondingly safety for \eqref{eq:controlaffinesystem}. And since the constraints are affine and thus convex, the feasible $\bar{u}(t_k) $ can be found with the QP.
	
	Recursive feasibility can be seen directly from step 4 in the Tube-CBF algorithm. Since the nominal state $\bar{x}(t_k) \in \mathcal{C}'$ for all $k\geq 0$ and since \eqref{eq:tubeCBFRTI} is feasible for all $\bar{x} \in \mathcal{C}'$, we see that \eqref{eq:tubeCBFRTI} is recursively feasible for all time steps $t_k$ with $k \geq 0$.
\end{proof}

\section{Results}\label{sec:results}

To demonstrate the performance of the proposed algorithms DBC and Tube-CBF, we performed simulations and hardware experiments. We simulated the DBC as a safety filter in conjunction with a nominal LQR controller, and we implemented the Tube-CBF algorithm with a nominal RTI on a mini-Segway and compared it to several baseline controllers. Figure \ref{fig:mini-segway_1} shows the mini-Segway, which is equipped with an Arduino capable ATmega32U4 MCU, wheel encoders, low-level motor controllers, and an LSM6DS33 IMU. The mini-Segway connects to a Raspberry~Pi model 3B+ through I2C which runs Ubuntu 18.04 and performs the controller computations. We only run the system in a planar mode, i.e., apply the same input to both wheels. The publication \cite{kim2015dynamic} presents the dynamics of the mini-Segway.

\subsection{Discrete Barrier Condition}

Although the DBC provides a very general formulation for safety guarantees, the approach is too conservative for the particularly fast dynamics of our hardware. Specifically, the resulting Lipschitz constants are large and consequently the feasible region has collapsed.

However, the DBC is applicable for dynamics with smaller Lipschitz constants. To show the feasibility of the DBC we performed simulations for two different mini-Segways. One has the same dynamics as our hardware, whereas the other has slower dynamics and consequently smaller Lipschitz constants. We use an LQR controller in conjunction with the DBC safety filter \eqref{eq:affconditionsafetyfilter} to ensure safety. We simulate a regulation task on the position and apply the safety filter at various rates. Without the filter, the system is unsafe for both the fast and slow dynamics. For sufficiently large filter rates and slow dynamics, the DBC is viable and successfully keeps the system safe. Figure \ref{fig:simulationone} shows the simulation results.


\begin{figure}[t]
	\centering
	\includegraphics[width=0.35\textwidth]{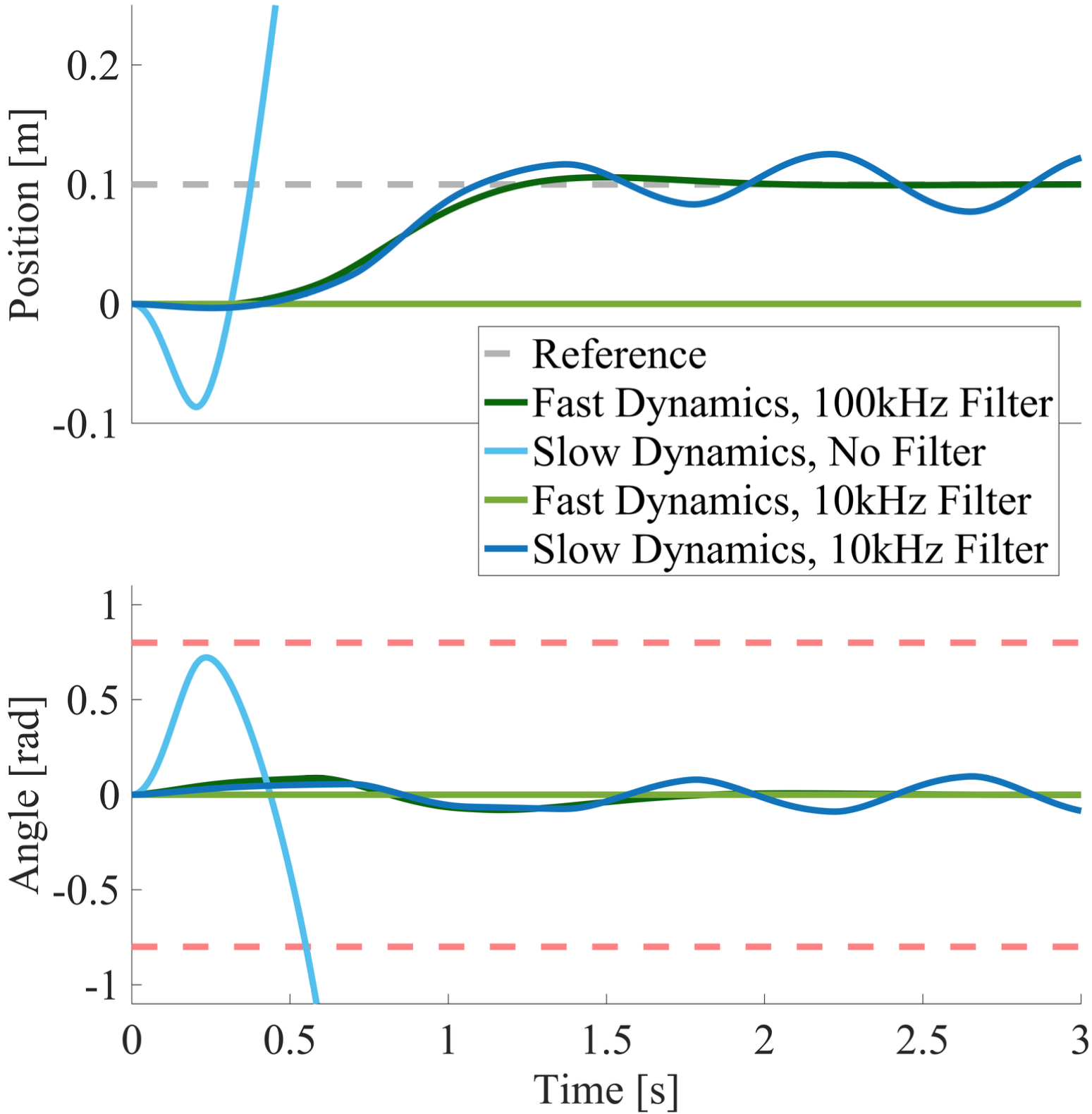}
	\caption{
		The figure shows the simulation of two mini-Segways with different dynamics . We use an LQR controller in conjunction with a DBC safety filter \eqref{eq:affconditionsafetyfilter} to keep the system safe. We simulate a regulation task with the safety filter applied at various rates. Without the filter, the systems are unsafe (only slow unsafe dynamics are shown). For sufficiently large filter rates and slow dynamics, the DBC is viable and successfully keeps the system safe.}
	\label{fig:simulationone}
\end{figure}

\subsection{Tube-CBF Implementation on Hardware}

\begin{figure}[t]
	\centering
	\includegraphics[width=0.33\textwidth]{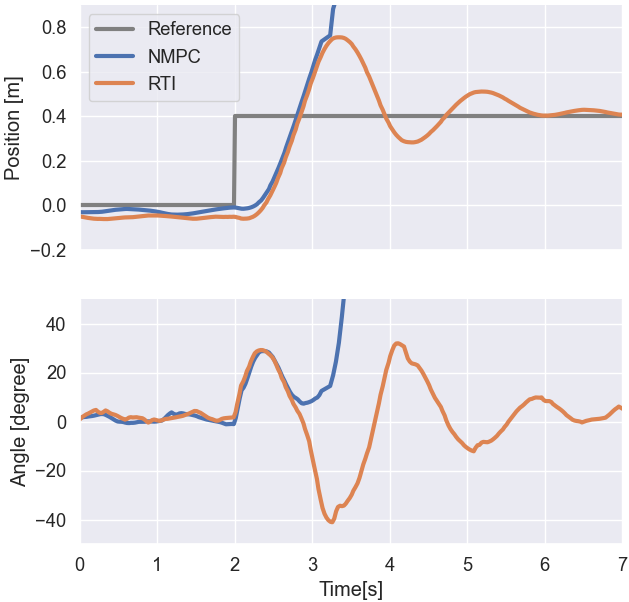}
	\caption{The figure shows a full NMPC controller with $N = 15$ and a plain RTI controller with $N = 50$, both applied at $\SI{33}{\hertz}$. At $t = \SI{2}{\second}$ the reference position changes from $\SI{0}{\meter}$ to $\SI{0.4}{\meter}$. While the full NMPC controller violates the safe set, the RTI remains safe due to its longer prediction horizon.}
	\label{fig:firstexperiment}
\end{figure}

We implement the proposed combination of RTI with Tube-CBF on the mini-Segway and compare its performance to three baseline controllers:
\begin{itemize}
	\item Full NMPC: We tackle the NLP \eqref{eq:NLP} with an interior-point solver
	\item Plain RTI: RTI without CBFs
	\item RTI with CBF: We solve the QP \eqref{eq:tubeCBFRTI} with the affine CBF condition, but do not use the Tube-CBF algorithm
\end{itemize}

\subsection{Derivation of Tube-CBF}
The only hard constraint present in our system is the motor input voltage of  $\pm \SI{5.4}{\volt}$. First, we design an auxiliary controller. We linearize the mini-Segway's dynamics around the origin and compute a full state feedback controller using LQR, i.e.,

\begin{equation}
	K_{aux} = 
	\AverageSmallMatrix{
		-3.7304 & 
		-3.4806 & 
		-0.7343
	} \nonumber
\end{equation}

We need the invariant tube $\Omega$ for $K_{aux}$ to perform the constraint tightening to arrive at a reduced safe set $\mathcal{C}'$ and the corresponding input constraints $\mathcal{U}'$. While \cite{yu2013tube} presents and approach to compute $\Omega$ using Lipschitz continuity, our system's closed-loop Lipschitz constants are large and the approach in \cite{yu2013tube} is too conservative. Thus, we opted for a \textit{not exact} constraint tightening of $\SI{33}{\percent}$ or $\pm \SI{1.8}{\volt}$, which we found to work well in practice.

Now that we have an auxiliary controller and the corresponding constraint tightening, we compute a CBF for $\mathcal{C}'$. To this end, we perform a numerical Hamilton-Jacobi reachability-analysis \cite{mitchell2005time} on a $99 \times 99 \times 99$ sampling grid of our state space $x = [\dot{s}, \theta, \dot{\theta}]^\textsf{T}$ considering our tightened motor voltage constraints of $\pm \SI{3.6}{\volt}$. We turned the sampled data into an analytical representation of a CBF using polynomial regression.


\subsection{Controller Design}

All RTI and NMPC controllers are designed with the same dynamics and a sampling time of $\SI{70}{\milli\s}$ using FORCES PRO \cite{FORCESPro}, with a quadratic cost of the form

\begin{equation}
	l(x,u) = \lVert \Delta u \rVert_{R_1} + \lVert u \rVert_{R_2} + \lVert x \rVert_{Q}. \nonumber
\end{equation}

When it comes to computational performance, there is a trade-off between prediction horizon length and update frequency. A longer horizon increases computational intensity, as does a higher frequency. We found the minimal frequency required for stability to be around $\SI{33}{\hertz}$. The following configurations exploit all of the available computational resources.

\begin{figure}[t]
	\centering
	\includegraphics[width=0.33\textwidth]{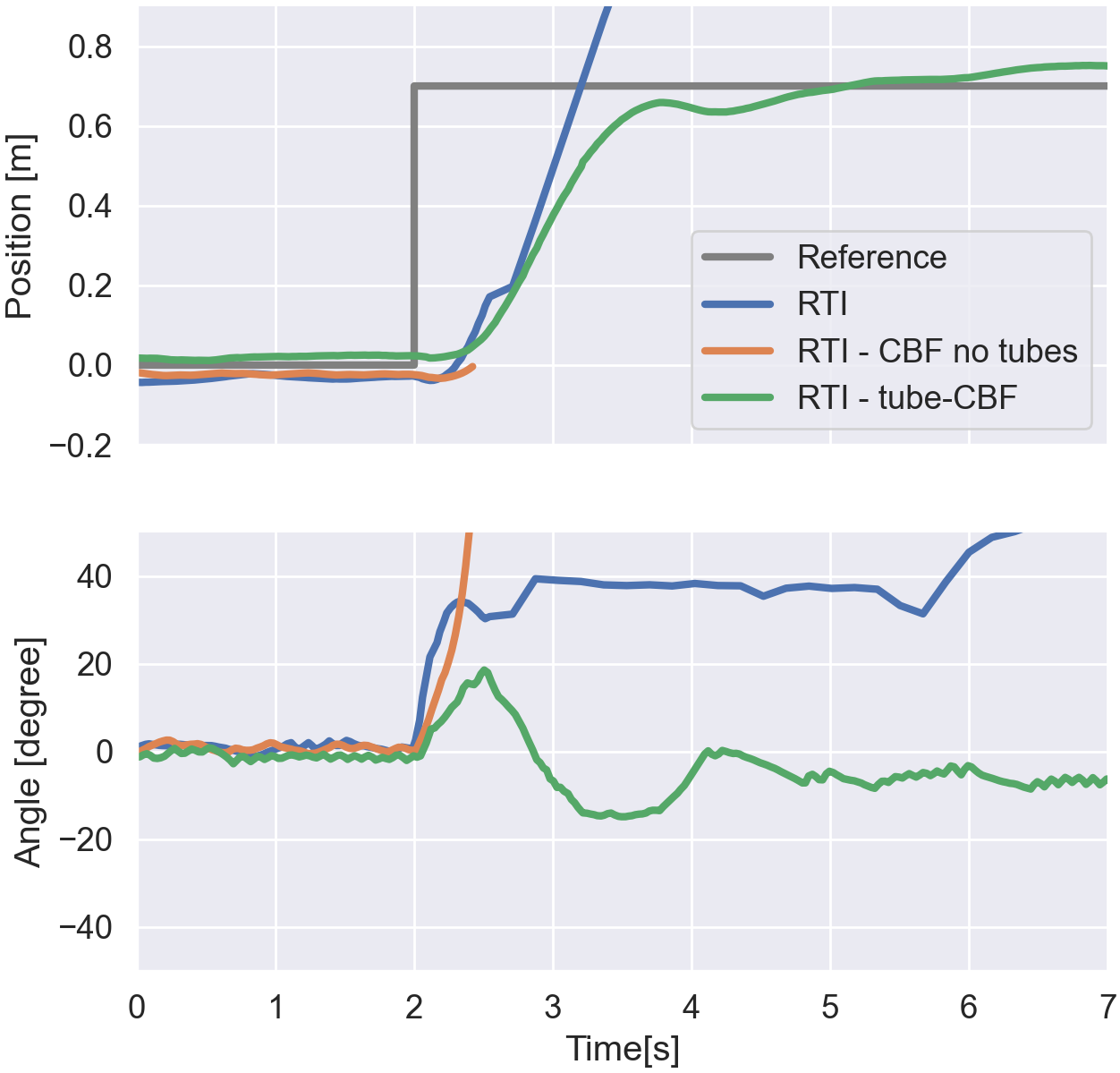}
	\caption{The figure shows  three RTI controllers, one with $N = 50$, applied at $\SI{33}{\hertz}$ and two with $N = 15$ applied at $\SI{100}{\hertz}$, where one uses Tube-CBF and the other uses the affine CBF condition without the Tube-CBF algorithm. At $t = \SI{2}{\second}$ the reference position changes from $\SI{0}{\meter}$ to $\SI{0.7}{\meter}$. Only the RTI controller with Tube-CBF remains safe.}
	\label{fig:secondexperiment}
\end{figure}

\begin{table}[h!]
	\begin{center}
		\label{tab:modelquantities}
		\begin{tabular}{l|c|c} 
			\textbf{Controller} & \textbf{Horizon Length} & \textbf{Frequency} \\
			\hline
			Full NMPC & $15$ & $\SI{33}{\hertz}$\\
			RTI & $50$ & $\SI{33}{\hertz}$ \\
			RTI - with CBF & $15$ & $\SI{100}{\hertz}$ \\
			RTI - with Tube-CBF & $15$ & $\SI{100}{\hertz}$ \\
		\end{tabular}
	\end{center}
\end{table}

\subsection{Control Task and Results}
The experimental task is to track a reference position. The results of the experiments are captured in figures \ref{fig:firstexperiment} and \ref{fig:secondexperiment}. Figure \ref{fig:firstexperiment} shows the performance of the full NMPC and the RTI. While the full NMPC cannot handle a $\SI{0.4}{\meter}$ step in the reference position, the RTI controller remains safe due to its longer prediction horizon. However, increasing the step in the reference position from $\SI{0.4}{\meter}$  to $\SI{0.7}{\meter}$ renders also the RTI controller unsafe, as shown in figure \ref{fig:secondexperiment}. Extending the RTI controller with our Tube-CBF algorithm successfully keeps the system safe. It's noteworthy that just extending the RTI formulation with the affine CBF condition but without applying the Tube-CBF algorithm, does not lead to safety. In fact, the affine CBF condition becomes infeasible shortly after the step in the reference position.

Thus, this shows that extending MPC with CBFs enhances safety. When using the affine CBF condition, compensating the discretization error is necessary, for example with the Tube-CBF algorithm.

\section{Conclusion}
In this work, we extend CBFs to account for discretization error and guarantee constraints satisfaction for nonlinear control-affine systems under discrete-time nominal controllers. We presented two algorithms that result in enforcing (a set of) affine conditions at a finite rate to guarantee safety. The DBC relies on Lipschitz continuity to capture the discretization error. The Tube-CBF algorithm relies on an auxiliary controller that spans an invariant tube to compensate for arising discretization errors. We combine Tube-CBF and approximate NMPC to obtain an efficient algorithm with continuous-time safety guarantees despite the approximate nature of the controller. We validate our proposed approach in hardware experiments on a mini-Segway, and demonstrate the need to account for discretization errors to guarantee safety at all times given limited computational resources.

Future work will focus on making auxiliary controllers and robust control invariant sets available for a larger number of systems. Furthermore, we will extend our formulation to include more real-world challenges, like external disturbances and model mismatch to provide more realistic implementations.



%


\bibliographystyle{IEEEtran}
\bibliography{bibliography}


\end{document}